\documentclass[english,english,preprint,aps,prl,superscriptaddress,nobibnotes]{revtex4-1}
\usepackage[T1]{fontenc}
\usepackage[latin9]{inputenc}
\usepackage{amsmath}
\usepackage{amsthm}
\usepackage{amssymb}
\usepackage{graphicx}

\makeatletter
\theoremstyle{plain}
\newtheorem{lem}{\protect\lemmaname}
\ifx\proof\undefined\
  \newenvironment{proof}[1][\proofname]{\par
    \normalfont\topsep6\p@\@plus6\p@\relax
    \trivlist
    \itemindent\parindent
    \item[\hskip\labelsep
          \scshape
      #1]\ignorespaces
  }{%
    \endtrivlist\@endpefalse
  }
  \providecommand{\proofname}{Proof}
\fi
\theoremstyle{plain}
\newtheorem{thm}{\protect\theoremname}

\usepackage{latexsym}
\usepackage{bm}
\usepackage{amsfonts}
\usepackage[normalem]{ulem}
\usepackage{enumitem}
\usepackage{amsthm}

\usepackage[hypertexnames=false]{hyperref}
\hypersetup{
	breaklinks = true,
    colorlinks = true,
    citecolor = {blue},
	urlcolor = {blue},
	linkcolor = {blue}
}

\makeatother

\usepackage{babel}
\providecommand{\lemmaname}{Lemma}
\providecommand{\theoremname}{Theorem}

\begin{document}
\title{The Physics of Spontaneous Parity-Time Symmetry Breaking in the Kelvin-Helmholtz Instability}
\author{Yichen Fu}
\author{Hong Qin}
\email{hongqin@princeton.edu}

\affiliation{Princeton Plasma Physics Laboratory, Princeton, NJ, 08543, USA}
\begin{abstract}
We show that the dynamics, in particular the Kelvin-Helmholtz (KH) instability, of an inviscid fluid with velocity shear admits Parity-Time (PT) symmetry, which provides a physical explanation to the well-known observation that the spectrum of the perturbation eigenmodes of the system is symmetric with respect to the real axis. It is found that the KH instability is triggered when and only when the PT symmetry is spontaneously broken. The analysis of PT symmetry also reveals that the relative phase between parallel velocity and pressure perturbations needs to be locked at $\pi/2$ when the instability is suppressed. 
\end{abstract}
\maketitle
A keen interest in parity-time (PT)-symmetric systems was initiated in 1998 when Bender et al. \citep{bender1998real} found that a large class of non-Hermitian Hamiltonians $\mathcal{H}$ exhibits entirely real spectrum, provided that the Hamiltonians have unbroken PT symmetry. Based on this observation, the requirement that a quantum operator corresponding to an observable must be Hermitian could be relaxed to a more physical condition of being PT-symmetric \citep{bender2002complex,bender2003must}. After being discovered in optical systems both theoretically \citep{makris2008beam} and experimentally \citep{guo2009observation}, PT symmetry has been wildly studied in many branches of quantum physics \citep{peng2016anti,gao2015observation,kreibich2014realizing} and classical physics \citep{auregan2017p,bender2013observation,bittner2012p,fleury2015invisible,qin2019kelvin}.

Kelvin-Helmholtz (KH) instability is a fluid instability occurring at the interface of two fluids with different velocities, or in a single fluid with velocity shear, i.e., a shear layer. In the classical theory of the KH instability, it is known that \citep{drazin1966hydrodynamic} the spectrum of the perturbation eigenmodes of an inviscid shear layer is symmetric with respect to the real axis. If $\omega$ is the eigen-frequency of the system, so is its complex conjugate $\bar{\omega}$. However, for a viscous shear layer, the eigenmode spectrum is usually not symmetric with respect to the real axis \citep{mack1976numerical}. The physical reason for these two different behaviors has not been identified. We demonstrate here that whether the spectrum is symmetric with respect to the real axis is determined by whether the system admits PT symmetry. For an inviscid shear layer, we show that the dynamics is PT-symmetric and thus its spectrum is symmetric with respect to the real axis. Identifying the PT symmetry admitted by the system reveals that the KH instability is triggered when and only when the PT symmetry is spontaneously broken. Furthermore, the analysis of PT symmetry also reveals that the relative phase between parallel velocity and pressure perturbations needs to be locked at $\pi/2$ when the instability is suppressed. In contrast, the governing equation of a viscous shear layer is not PT-symmetric, or PT symmetry is explicitly broken \citep{brading2003symmetries}, and this is the physical reason why its spectrum is not symmetric with respect to the real axis.

To begin our discussion, we first provide a brief derivation of the governing equations of the KH instability. Starting from Rayleigh \citep{rayleigh1879stability}, the stability of shear layers has been studied with different shear velocity, temperature, viscosity and boundary conditions. Following Blumen and Drazin et al. \citep{blumen1970shear,blumen1975shear,drazin1977shear}, we consider the linear stability of a two-dimensional compressible inviscid shear layer with a uniform temperature. The system is governed by the fluid equations, 
\begin{align}
\partial_{t}\rho+\nabla\cdot(\rho\boldsymbol{v}) & =0,\label{mass_eq}\\
\rho\left(\partial_{t}+\boldsymbol{v}\cdot\nabla\right)\boldsymbol{v} & =-\nabla p,\label{momentum_eq}\\
\left(\partial_{t}+\boldsymbol{v}\cdot\nabla\right)\left(p/\rho^{\gamma}\right) & =0,\label{state_eq}
\end{align}
where $\rho$ and $p$ are the density and pressure of the fluid, $\boldsymbol{v}=(v_{x},v_{z})$ is the velocity field in the $(x,z)$ plane, and $\gamma$ is the ratio of specific heats.

The unperturbed equilibrium flow is $\boldsymbol{v}_{0}=(v_{0}(z),0)$, a shear flow moving in the $\hat{\boldsymbol{x}}$ direction with variation along the $\hat{\boldsymbol{z}}$ direction. The equilibrium pressure $p=p_{0}$ and equilibrium density $\rho=\rho_{0}$ are assumed to be constant. Consider a linear perturbation of the form 
\begin{align}
\begin{split}v_{x} & =v_{0}(z)+v_{1x}(t,x,z),\quad v_{z}=v_{1z}(t,x,z),\\
\rho & =\rho_{0}+\rho_{1}(t,x,z),\quad p=p_{0}+p_{1}(t,x,z).
\end{split}
\end{align}
Assuming $(v_{1x},v_{1z},p_{1})\sim\exp[\mathrm{i}(kx-\omega t)]$ for a wavelength $k\in\mathbb{R}$, the linearized equations can be written as 
\begin{align}
\mathrm{i}\omega\begin{pmatrix}v_{1x}\\
v_{1z}\\
p_{1}
\end{pmatrix}=\begin{pmatrix}\mathrm{i}kv_{0} & v_{0}' & \mathrm{i}k\\
0 & \mathrm{i}kv_{0} & \partial_{z}\\[5pt]
\dfrac{\mathrm{i}k}{M^{2}} & \dfrac{1}{M^{2}}\partial_{z} & \mathrm{i}kv_{0}
\end{pmatrix}\begin{pmatrix}v_{1x}\\
v_{1z}\\
p_{1}
\end{pmatrix},\label{linearized}
\end{align}
where $'$ denotes $\mathrm{d}/\mathrm{d}z$, and variables of the system have been normalized by the characteristic velocity $U$ and the characteristic scale length $L$. In Eq.\,(\ref{linearized}), $M\equiv U/\sqrt{\gamma p_{0}/\rho_{0}}$ is the Mach number. The boundary condition is $v_{1z}=p_{1}'=0$ at $z=\infty$ for unbounded fluids, or at $z=z_{\mathrm{b}}$ for fluids limited by rigid boundaries. When $M\to0$, the third row in Eq.\,(\ref{linearized}) recovers the incompressible equation $\nabla\cdot\boldsymbol{v}_{1}=0$ as a special case. The temporal stability of the shear layer is determined by solving the eigenvalue problem of Eq.\,(\ref{linearized}) with given $M,k$ and proper boundary conditions. If there exists an eigen-frequency $\omega=\omega_{\mathrm{R}}+\mathrm{i}\omega_{\mathrm{I}}$ with $\omega_{\mathrm{I}}>0$, then the KH instability occurs, and the profile $v_{0}(z)$ is unstable.

Equation (\ref{linearized}) assumes the form of the Schr\"{o}dinger equation $\mathcal{H}\psi=\omega\psi$ with 
\begin{align}
 & \mathcal{H}=\begin{pmatrix}kv_{0} & -\mathrm{i}v_{0}' & k\\
0 & kv_{0} & -\mathrm{i}\partial_{z}\\[5pt]
\dfrac{k}{M^{2}} & \dfrac{-\mathrm{i}}{M^{2}}\partial_{z} & kv_{0}
\end{pmatrix},\label{H-1}\\
 & \psi=(v_{1x},v_{1z},p_{1})^{\mathrm{T}}.
\end{align}
However, the Hamiltonian $\mathcal{H}$ here is not Hermitian. As we will see, $\mathcal{H}$ is PT-symmetric instead.

Now we develop the mathematical tools for analyzing the physics of PT symmetry contained in Eq.\,(\ref{H-1}). In general, a Hamiltonian $\mathcal{H}$ is PT-symmetric if 
\begin{align}
\big(\mathcal{P}\mathcal{T}\big)\mathcal{H}\big(\mathcal{P}\mathcal{T}\big)=\mathcal{H},\label{pt_condition}
\end{align}
where $\mathcal{P}$ is a linear operator satisfying $\mathcal{P}^{2}=1$, and $\mathcal{T}$ is the complex conjugate operator \citep{bender2007making}. In previous studies, two types of PT-symmetric Hamiltonians have been extensively studied. The first type consists of scalar differential operators \citep{bender1998real,makris2008beam,wadati2008construction,bender1999pt}, e.g., $\mathcal{H}=\partial_{x}^{2}+(\mathrm{i}x)^{N}$. The second type is linear maps of finite-dimensional complex vector spaces, which can be expressed as finite-dimensional square matrices \citep{el2007theory,guo2009observation,bender2013observation,qin2019kelvin}, for example, 
\begin{equation}
\mathcal{H}=\begin{pmatrix}a+\mathrm{i}b & g\\
g & a-\mathrm{i}b
\end{pmatrix}.
\end{equation}

The Hamiltonian operator (\ref{H-1}) studied here represents a new type, which takes the form of matrix differential operators, i.e., matrices whose elements are differential operators. This type of Hamiltonians was encountered in the study of Bose-Einstein condensates \citep{kartashov2014symmetric}, where Pauli matrices were used to describe the spin degree of freedom. An $m$-th order matrix differential operator can be written as 
\begin{align}
\mathcal{H}=\sum_{i=0}^{m}N_{i}(z)\dfrac{\partial^{i}}{\partial z^{i}},\label{matrix_differential_operator}
\end{align}
where $N_{i}$ are $n\times n$ complex matrices. The corresponding state is $\psi(z)=(\psi_{1},\cdots,\psi_{n})^{\mathrm{T}}$. This new type is a marriage between the two types of Hamiltonians described above. When $N_{i}$ are $1\times1$ matrices, $\mathcal{H}$ reduces to the first type described above; when $m=0$, $\mathcal{H}$ reduces to the second type.

For a Hamiltonian of the form of Eq.\,(\ref{matrix_differential_operator}), the PT-symmetry condition (\ref{pt_condition}) becomes 
\begin{align}
\begin{split} & \sum_{i=0}^{m}\Big[\big(\mathcal{P}\mathcal{T}\big)N_{i}\big(\mathcal{P}\mathcal{T}\big)\Big]\Big[\big(\mathcal{P}\mathcal{T}\big)\dfrac{\partial^{i}}{\partial z^{i}}\big(\mathcal{P}\mathcal{T}\big)\Big]=\sum_{i=0}^{m}N_{i}\dfrac{\partial^{i}}{\partial z^{i}}.\end{split}
\label{MDO_pt_condition}
\end{align}
If the parity operator $\mathcal{P}$ is an $n\times n$ constant matrix in $\mathbb{C}^{n}$, then the $\mathcal{P}\mathcal{T}$ operator commutes with the differential operators $\partial^{i}/\partial z^{i}$, and Eq.\,(\ref{MDO_pt_condition}) can be simplified to 
\begin{align}
\big(\mathcal{P}\mathcal{T}\big)N_{i}\big(\mathcal{P}\mathcal{T}\big)=N_{i},\quad i=0,1,\cdots,m.
\end{align}
This means that matrix differential operator $\mathcal{H}$ is PT-symmetric if all matrices $N_{i}$ are PT-symmetric with respect to the same $\mathcal{P}\mathcal{T}$ operator. We will focus on this special case in the present study.

The most important and well-known properties of a PT-symmetric Hamiltonian $\mathcal{H}$ \citep{bender2010pt,bender2007making} can be summarized as follows. (i) Its spectrum is symmetric with respect to the real axis, i.e., if $\omega$ is an eigenvalue of $\mathcal{H}$, so is its complex conjugate $\bar{\omega}$. (ii) If every eigenfunction $\psi$ of $\mathcal{H}$ is also an eigenfunction of the $\mathcal{P}\mathcal{T}$ operator, i.e., $\mathcal{P}\mathcal{T}\psi=\lambda\psi$ for some $\lambda$, then we say that PT symmetry is unbroken. In this case, the spectrum of $\mathcal{H}$ is always real. (iii) If some eigenfunctions of $\mathcal{H}$ are not the eigenfunctions of the $\mathcal{P}\mathcal{T}$ operator, we say that PT symmetry is spontaneously broken. The Hamiltonian $\mathcal{H}$ has a complex eigenvalue with $\omega_{\mathrm{I}}\neq0$ when and only when PT symmetry is spontaneously broken. In this case, there must exist an unstable eigenmode with $\omega_{I}>0$ due to the symmetry property of the spectrum.

Now we prove an important result that if the Hamiltonian specified by the matrix differential operator (\ref{matrix_differential_operator}) admits PT symmetry, then under proper change of basis for the state vectors the matrices $N_{i}$ can always be transformed into real matrices $\tilde{N}_{i}$. First, we need the following lemma adapted from Wigner's theory on normal forms for anti-unitary operators \citep{wigner1960normal}. For self-consistency and future reference, a complete proof is given here. 
\begin{lem}
\label{lemma} For an anti-unitary operator $A$ defined in a finite-dimensional complex Hilbert space, if $A^{2}=1$, then it is always possible to construct a set of complete orthonormal basis $\{\alpha_{i}\}$ satisfying $A\alpha_{i}=\alpha_{i}$. 
\end{lem}
\begin{proof}
First, select an arbitrary unit vector $\beta_{1}$ in the Hilbert space. If $A\beta_{1}\neq-\beta_{1}$, define a unit vector $\alpha_{1}=c_{1}(\beta_{1}+A\beta_{1})$, where $c_{1}$ is a real normalization constant. We have
\begin{align}
A\alpha_{1}=Ac_{1}(\beta_{1}+A\beta_{1})=c_{1}(A\beta_{1}+\beta_{1})=\alpha_{1}.
\end{align}
If $A\beta_{1}=-\beta_{1}$, define $\alpha_{1}=\mathrm{i}\beta_{1}$ and we have: 
\begin{align}
A\alpha_{1}=A\,\mathrm{i}\beta_{1}=-\mathrm{i}A\beta_{1}=\mathrm{i}\beta_{1}=\alpha_{1}.
\end{align}
Next, choose another unit vector in $\beta_{2}$ that is orthogonal to $\alpha_{1}$, i.e., $\beta_{2}^{\dagger}\alpha_{1}=0$, where $\dagger$ is the conjugate transpose operation. Construct a unit vector $\alpha_{2}$ from $\beta_{2}$ using the same procedure above. If $A\beta_{2}\neq-\beta_{2}$ then let $\alpha_{2}=c_{2}(\beta_{2}+A\beta_{2})$, and 
\begin{align}
\begin{split}\alpha_{1}^{\dagger}\alpha_{2} & =c_{2}\alpha_{1}^{\dagger}(\beta_{2}+A\beta_{2})=c_{2}\alpha_{1}^{\dagger}A\beta_{2}\\
 & =c_{2}\overline{(A\alpha_{1})^{\dagger}A^{2}\beta_{2}}=c_{2}\beta_{2}^{\dagger}\alpha_{1}=0,
\end{split}
\end{align}
where use is made of the anti-unitarity of $A$ in the third equal sign. If $A\beta_{2}=-\beta_{2}$ then let $\alpha_{2}=i\beta_{2}$, and obviously $\alpha_{1}^{\dagger}\alpha_{2}=0$. Thus, $\alpha_{2}$ is orthogonal to $\alpha_{1}$. The amplitude of $\alpha_{2}$ can be normalized to one by multiplying a real constant. Repeating the same procedure described above, we can build a set of complete orthonormal basis $\{\alpha_{i}\}$ satisfying $A\alpha_{i}=\alpha_{i}$. 
\end{proof}
\begin{thm}
\label{theorem} For a matrix differential operator $\mathcal{H}$ specified by Eq.\,(\ref{matrix_differential_operator}), if $[N_{i},\mathcal{P}\mathcal{T}]=0$ for all $i$, then the following statements hold.

(a) There exists a unitary matrix $O$ in $\mathbb{C}^{n}$ such that in terms of the transformed state vector $\phi=O\psi$, the Schr\"{o}dinger equation is $\tilde{\mathcal{H}}\phi=\omega\phi$ with $\tilde{\mathcal{H}}=\sum_{i=0}^{m}\tilde{N}_{i}\dfrac{\partial^{i}}{\partial z^{i}}$, and $\tilde{N}_{i}=ON_{i}O^{\dagger}$ are real matrices.

(b) The eigenvalue system $\tilde{\mathcal{H}}\phi=\omega\phi$ can be reduced to one ODE, 
\begin{align}
\sum_{i=1}^{n}a_{i}(\omega)\dfrac{\mathrm{d}^{i}}{\mathrm{d}z^{i}}\phi_{l}=0,\label{real_secular_ode}
\end{align}
in terms of one component $\phi_{l}$ of the state vector $\phi$, and the coefficients $a_{i}(\omega)$ are real-value functions in the sense that $a_{i}(\omega)\in\mathbb{R}$ for $\omega\in\mathbb{R}$, or equivalently, 
\begin{align}
a_{i}(\bar{\omega})=\overline{a_{i}(\omega)}.\label{coefficient_condition}
\end{align}
\end{thm}
\begin{proof}
Notice that $\mathcal{P}\mathcal{T}$ is an anti-unitary operator in $\mathbb{C}^{n}$ and $(\mathcal{P}\mathcal{T})^{2}=1$. According to Lemma \ref{lemma}, we can choose a set of complete orthonormal basis $\{X_{i}\}$ for $\mathbb{C}^{n}$ satisfying $\mathcal{P}\mathcal{T}X_{i}=X_{i}$. Here, each $X_{i}$ is a constant column vector in $\mathbb{C}^{n}$ and $X_{i}^{\dagger}X_{j}=\delta_{ij}$. Let $O\equiv(X_{1},\cdots,X_{n})^{\dagger},$ which belongs to $U(n).$ Let the new state vector is $\phi=O\psi$. In terms of $\phi,$ the Schr\"{o}dinger equation is 
\begin{align}
\begin{split}\omega\phi & =\tilde{\mathcal{H}}\phi,\\
\tilde{\mathcal{H}} & \equiv O\left[\sum_{i=0}^{m}N_{i}\dfrac{\partial^{i}}{\partial z^{i}}\right]O^{\dagger}=\sum_{i=0}^{m}\tilde{N}_{i}\dfrac{\partial^{i}}{\partial z^{i}},\\
\tilde{N}_{i} & \equiv ON_{i}O^{\dagger}.
\end{split}
\end{align}
Following the procedure in \citep{bender2002generalized}, we can prove that $\tilde{N}_{i}$ are real as follows. Because $\mathcal{P}\mathcal{T}$ is anti-unitary and $N_{i}$ commutes with $\mathcal{P}\mathcal{T}$, we have 
\begin{align}
\begin{split}\left(\tilde{N}_{i}\right)_{kl} & \equiv X_{k}^{\dagger}\,N_{i}\,X_{l}=\overline{\left(\mathcal{P}\mathcal{T}X_{k}\right)^{\dagger}\left(\mathcal{P}\mathcal{T}N_{i}X_{l}\right)}\\
 & =\overline{\left(\mathcal{P}\mathcal{T}X_{k}\right)^{\dagger}\left(N_{i}\mathcal{P}\mathcal{T}X_{l}\right)}=\overline{X_{k}^{\dagger}N_{i}X_{l}}\\
 & =\overline{\left(\tilde{N}_{i}\right)_{kl}}.
\end{split}
\end{align}
This proves part (a). Equation $\tilde{\mathcal{H}}\phi=\omega\phi$ consists of $n$ coupled linear differential equations with dependent variables $\phi=(\phi_{1},\cdots,\phi_{n})^{\mathrm{T}}$. In principle, we could eliminate $n-1$ components of $\phi$ in favor of one $\phi_{l}$, and the resulting governing equation assumes the form of Eq.\,(\ref{real_secular_ode}). Since $\tilde{N}_{i}$ are all real matrices, the coefficient $a_{i}(\omega)$ of the reduced ODE (\ref{real_secular_ode}) must be real-value functions.
\end{proof}
A corollary of Theorem \ref{theorem} is that the coefficients of the characteristic polynomials of $N_{i}$ are real. When $N_{i}=0\,\,(i>0),$ the Hamiltonian does not contain differential operators and can be represented by a complex matrix $N_{0}$. In this special case, Theorem 1 implies that the coefficients of the characteristic polynomial of $N_{0},$ which determines the spectrum of the system, are real \citep{bender2010pt}.

We now return to the governing system for the KH instability (\ref{linearized}). Its Hamiltonian (\ref{H-1}) contains two components, 
\begin{align}
 & \mathcal{H}\equiv N_{0}+N_{1}\dfrac{\partial}{\partial z},\label{H}\\
 & N_{0}=\begin{pmatrix}kv_{0} & -\mathrm{i}v_{0}' & k\\
0 & kv_{0} & 0\\
\dfrac{k}{M^{2}} & 0 & kv_{0}
\end{pmatrix},\quad N_{1}=\begin{pmatrix}0 & 0 & 0\\
0 & 0 & -\mathrm{i}\\
0 & \dfrac{-\mathrm{i}}{M^{2}} & 0
\end{pmatrix}.
\end{align}
It is easy to verify that $N_{0}$ and $N_{1}$ are PT-symmetric for $\mathcal{P}=\mathrm{diag}(1,-1,1)$, i.e., 
\begin{align}
(\mathcal{P}\mathcal{T})N_{i}(\mathcal{P}\mathcal{T})=\mathcal{P}\bar{N_{i}}\mathcal{P}=N_{i},\quad i=0,1.
\end{align}
Thus, $\mathcal{H}$ is indeed PT-symmetric. It follows that the spectrum of the system is symmetric with respect to the real axis, and the KH instability is triggered when and only when PT symmetry is spontaneously broken. An orthonormal basis of the $\mathcal{P}\mathcal{T}$ operator are $(1,0,0)^{\mathrm{T}}$, $(0,\mathrm{i},0)^{\mathrm{T}}$ and $(0,0,1)^{\mathrm{T}}$. According to the proof of Theorem \ref{theorem}, the unitary matrix generating the basis transformation for the state vector is 
\begin{equation}
O=\begin{pmatrix}1 & 0 & 0\\
0 & i & 0\\
0 & 0 & 1
\end{pmatrix},
\end{equation}
and the transformed Hamiltonian is 
\begin{align}
\tilde{\mathcal{H}} & =\tilde{N}_{0}+\tilde{N}_{1}\dfrac{\partial}{\partial z},\\
\tilde{N}_{0} & =\begin{pmatrix}kv_{0} & v_{0}' & k\\
0 & kv_{0} & 0\\
\dfrac{k}{M^{2}} & 0 & kv_{0}
\end{pmatrix},\quad\tilde{N_{1}}=\begin{pmatrix}0 & 0 & 0\\
0 & 0 & -1\\
0 & \dfrac{1}{M^{2}} & 0
\end{pmatrix}.
\end{align}
The matrices $\tilde{N}_{0}$ and $\tilde{N}_{1}$ are real, and the governing system can be reduced to a single ODE satisfying the coefficient condition (\ref{coefficient_condition}). Indeed, from Eq.\,(\ref{linearized}), it is straightforward to eliminate $v_{1x}$ and $p$ in favor of $v_{1z}$ to obtain one single second-order ODE,
\begin{align}
\left(\dfrac{(v_{0}-c)v_{1z}'-v_{0}'v_{1z}}{1-M^{2}(v_{0}-c)^{2}}\right)'-k^{2}(v_{0}-c)v_{1z}=0.\label{vz_ode}
\end{align}
Here, $c=\omega/k$ is the phase velocity. Obviously, condition (\ref{coefficient_condition}) is satisfied for Eq.\,(\ref{vz_ode}). At the $M\to0$ limit, this equation becomes the Rayleigh stability equation. Due to the term $(v_{0}-c)v_{1z}''$, the eigenvalue problem is singular. This leads to singular eigenmodes with continuous spectrum in addition to well-behaved eigenmodes with discrete spectrum \citep{drazin1966hydrodynamic,lees1946investigation}. The continuous spectrum locates at $c=v_{0}$ and is real with logarithmic divergence for the mode structure. The discrete spectrum could be either real or complex.

Historically, Eq.\,(\ref{vz_ode}) had been derived without the knowledge of PT symmetry. Because the coefficients of Eq.\,(\ref{vz_ode}) satisfy condition (\ref{coefficient_condition}), it had been known that the spectrum of the system is symmetric with respect to the real axis. However, the key question as to why the coefficients of Eq.\,(\ref{vz_ode}) satisfy condition (\ref{coefficient_condition}) was never asked. As we have shown here, this property is a consequence of the fact that the system admits PT symmetry.

The physics of PT symmetry can be understood from the perspective of the Lorentz group, the homogeneous symmetry of flat spacetime. The $\mathcal{P}\mathcal{T}$ transformation is an element in the Lorentz group $O(1,3),$ which as a manifold contains 4 connected components, each of which is topologically separated from others. One of the components is the proper, orthochronous Lorentz group $SO^{+}(1,3),$ and $O(1,3)$ is a semi-direct product of $SO^{+}(1,3)$ and the discrete subgroup $\{1,\mathcal{P},\mathcal{T},\mathcal{P}\mathcal{T}\}$, 
\[
O(1,3)=SO^{+}(1,3)\rtimes\{1,\mathcal{P},\mathcal{T},\mathcal{P}\mathcal{T}\},
\]
where $\mathcal{P}=\mathrm{diag}(1,-1,-1,-1)$ and $\mathcal{T}=\mathrm{diag}(-1,1,1,1).$ In general, we expect that physics is invariant with respect to $SO^{+}(1,3)$, but not with respect to $\mathcal{P}$ transformation or $\mathcal{T}$ transformation. The program initiated by Bender is to investigate the interesting physics associated with $\mathcal{P}\mathcal{T}$ transformation \citep{Bender2019}. It was demonstrated that in classical systems governed by Newton's second law, such as the dynamical systems in neutral fluids and plasmas, PT symmetry is a consequence of reversibility \citep{qin2019kelvin,Zhang2020}. When a system is not subject to any dissipation, the dynamics is reversible and admits PT symmetry. Note that this observation is consistent with Bender's characterization of PT symmetry as a mechanism of balanced grain and loss for two coupled subsystems \citep{bender2002complex,bender2002generalized,bender2003must,bender2007making,bender2010pt,bender2013observation}. If the loss of one subsystem is balanced with the gain of the other subsystem, then the whole system is free of dissipation.

For the KH instability investigated, if a viscosity term $\mu\nabla^{2}\boldsymbol{v}$ is included in Eq.\,(\ref{momentum_eq}), the coefficients of the reduced ODE do not satisfy the condition (\ref{coefficient_condition}). In the $M\to0$ limit, the reduced ODE is the Orr-Sommerfeld equation with complex coefficients \citep{criminale2018theory}. In these cases, the spectrum is not symmetric with respect to the real axis. From our analysis above, it is clear now that the physics here is that viscosity renders the system irreversible and explicitly breaks PT symmetry.

In addition to the properties of spectrum, PT-symmetry analysis also leads to more detailed information about the instability previously unknown. The condition of unbroken symmetry is 
\begin{align}
\mathcal{P}\mathcal{T}\begin{pmatrix}v_{1x}\\
v_{1z}\\
p_{1}
\end{pmatrix}=\begin{pmatrix}\bar{v}_{1x}\\
-\bar{v}_{1z}\\
\bar{p}_{1}
\end{pmatrix}=\lambda\begin{pmatrix}v_{1x}\\
v_{1z}\\
p_{1}
\end{pmatrix}
\end{align}
for some $\lambda.$ Therefore, 
\begin{align}
 & \bar{f}=f,\,\,\,\,\bar{g}=-g,\label{pt_condition2}
\end{align}
where $f=v_{1x}/p_{1}$ and $g=v_{1z}/p_{1}.$ Equation (\ref{pt_condition2}) requires that when the system has unbroken PT symmetry, $f$ is real and $g$ is imaginary. It means that if the system is stable, $v_{1z}$ should always have a $\pi/2$ phase difference relative to $v_{1x}$ and $p_{1}$. When the system is unstable, PT symmetry is spontaneously broken and $\mathcal{P}\mathcal{T}\psi=\lambda\psi$ does not hold. Thus, the phase differences between these components become arbitrary. Interestingly, similar effects were also observed in optical systems \citep{ruter2010observation}, where the PT-symmetric system consists of two coupled waveguides. In these experiments, when PT symmetry was unbroken, the phase difference between two waveguides could be an arbitrary value between $[0,\pi]$; when PT symmetry is broken, the phase difference was locked at $\pi/2$.

Now we proceed to demonstrate the PT symmetry and the breaking thereof by numerical examples of the eigenvalues and eigenfunctions of the Hamiltonian $\mathcal{H}$ defined in Eq.\,(\ref{H-1}). The velocity profile is taken to be $v_{0}(z)=\tanh(z)+1$ for $z\in[-2,4]$ \citep{michalke1964inviscid,blumen1970shear}. Rigid-wall boundary conditions are imposed at $z=-2$ and $z=4$. We scan Mach number $M$ from $0.05$ to $0.6$. At each Mach number, wave number $k$ varies from $0.1$ to $1.2$. For a given set of $M$ and $k$, the eigen-frequency $\omega$ and eigenmodes $\psi$ are numerically solved with boundary conditions $v_{1z}=p_{1}'=0$. Two numerical methods are used. For the discrete spectrum with well-behaved mode structure, the standard shooting method can be applied straightforwardly to Eq.\,(\ref{vz_ode}). To solve for the continuous spectrum with logarithmic divergent mode structure, a different algorithm was developed. Instead of solving Eq.\,(\ref{vz_ode}), the Hamiltonian operator $\mathcal{H}$ is discretized and the spectrum of the discretized Hamiltonian recovers that of $\mathcal{H}$ under proper limits. This algorithm is applicable to both well-behaved modes on the discrete spectrum and singular modes on the continuous spectrum \citep{Fu2020}. The numerically calculated stability diagram in the $M$-$k$ plane is shown in Fig.\,\ref{stability_boundary}. In the upper (green) region, the system is stable with unbroken PT symmetry, and in the lower (red) region, the system is unstable with spontaneously broken PT symmetry. For this problem, the system is stable with unbroken PT symmetry on the boundary between upper and lower regions.

\begin{figure}[ht]
\includegraphics[width=0.8\linewidth]{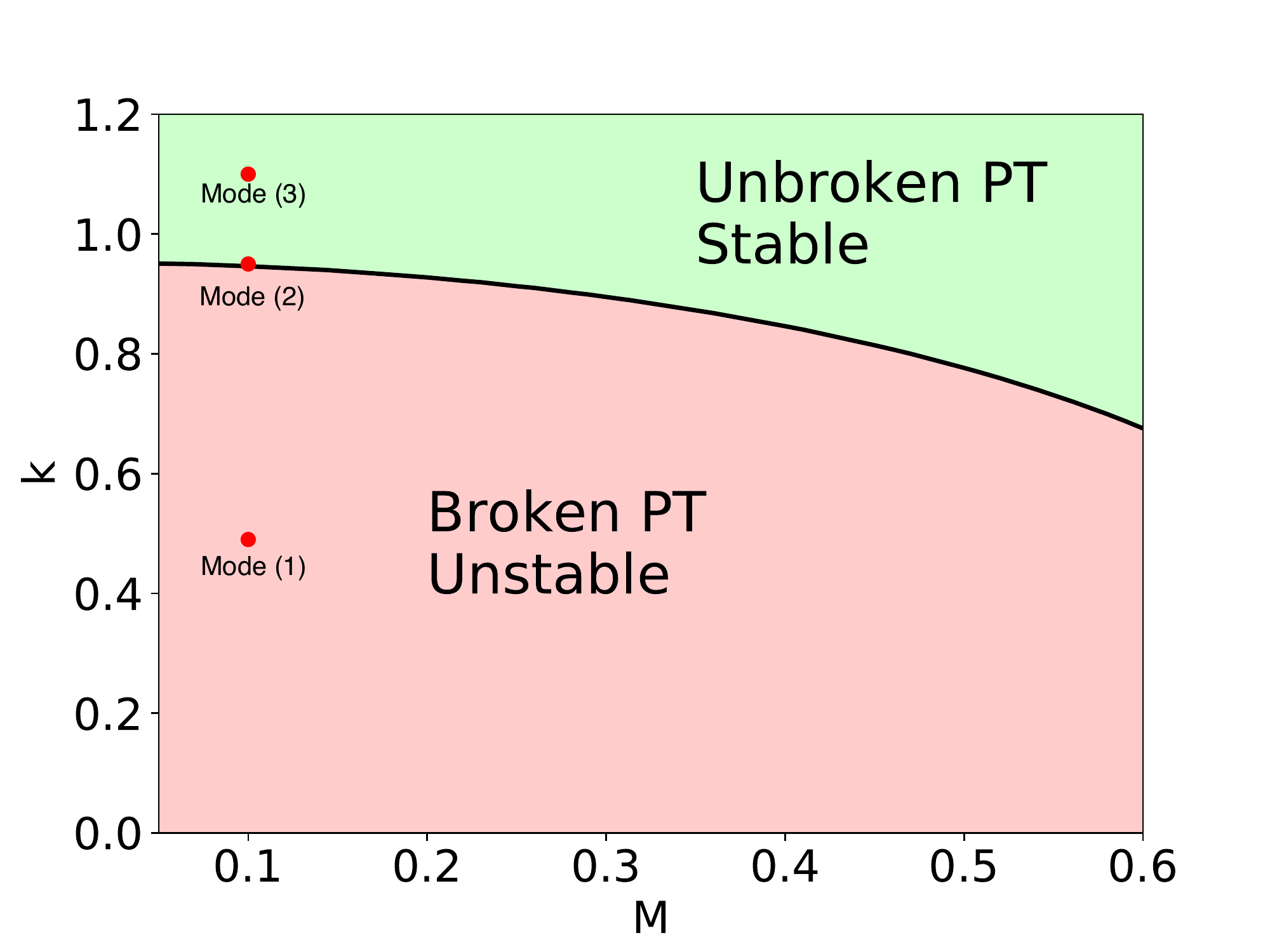} \caption{Stability diagram in the $M$-$k$ plane. Unstable modes are found in the region with spontaneously broken PT symmetry. All modes in the unbroken PT-symmetry region are stable.}
\label{stability_boundary}
\end{figure}

To explicitly verify PT-symmetry breaking as the mechanism for the KH instability, numerical solutions for three sets of parameters marked by red points in the $M$-$k$ plane are obtained. Mode (1) is an unstable mode on the discrete spectrum in the upper (green) region, Mode (2) is a stable mode on the discrete spectrum on the boundary between upper and lower region, and Mode (3) is a stable mode on the continuous spectrum in the lower (red) region. The corresponding functions $f$ and $g$ for each solution are shown in Fig.\,\ref{eigenfunctions}. Modes (2) and (3) are stable and have unbroken PT symmetry. Thus, $f$ is real and $g$ is imaginary as expected. Mode (1) is unstable and PT symmetry is spontaneously broken. Therefore, $f$ and $g$ are complex with both real and imaginary parts and vary as functions of $z$.

\begin{figure}[ht]
\includegraphics[width=0.9\linewidth]{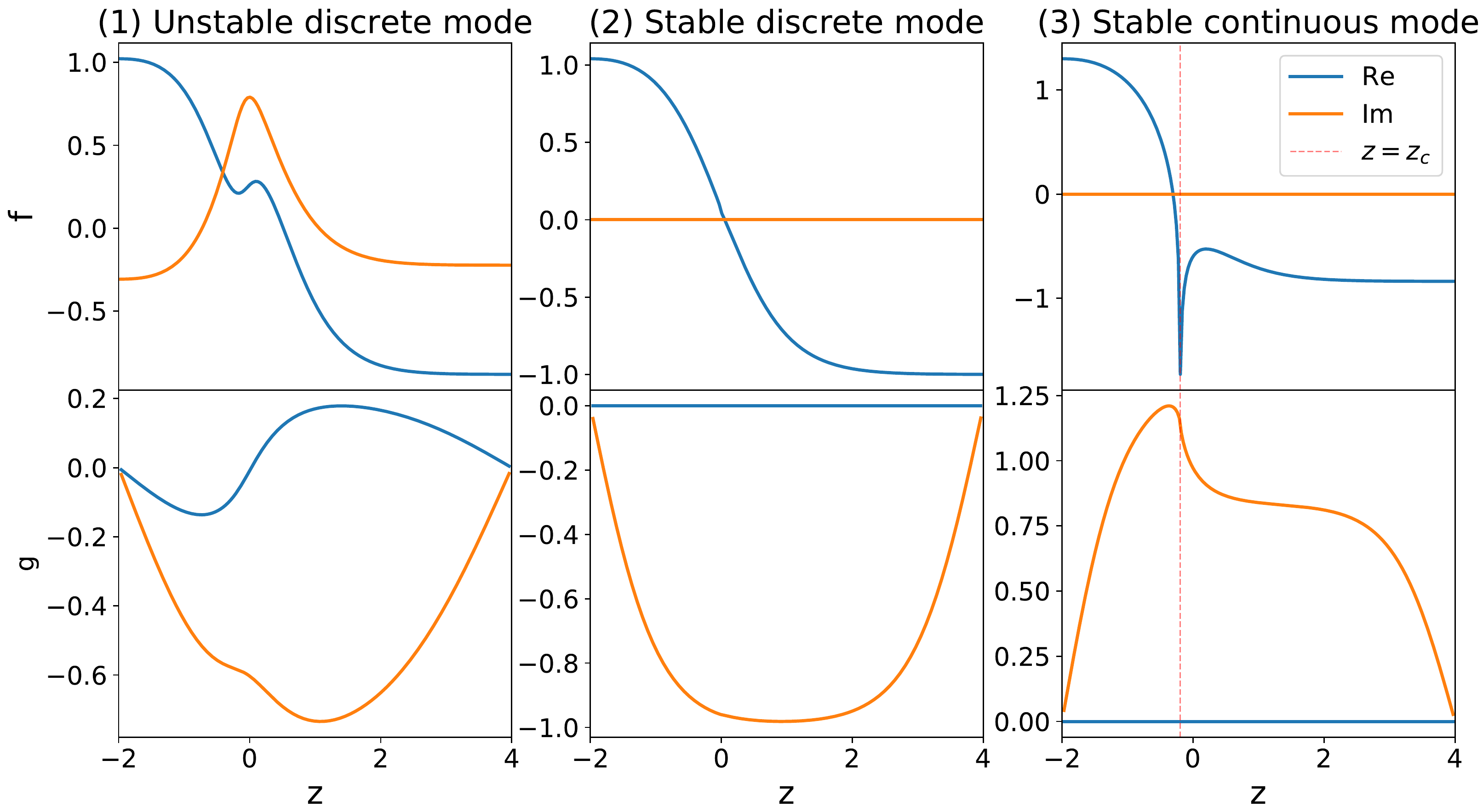} \caption{The imaginary and real parts of $f=v_{1x}/p_{1}$ and $g=v_{1z}/p_{1}$ for the three eigenmodes displayed in Fig.\,\ref{stability_boundary}. (1) $M=0.1$, $k=0.49$, and $\omega=0.457+0.132\mathrm{i}$; (2) $M=0.1$, $k=0.960$, and $\omega=0.958$; (3) $M=0.1$, $k=1.1$, and $\omega=0.885$. Mode (3) is on the continuous spectrum with a logarithmic divergence at $z_{c}=-0.198$.}
\label{eigenfunctions}
\end{figure}

In conclusion, together with \citep{qin2019kelvin}, we have proved that the KH instability is the result of spontaneous PT-symmetry breaking. The discovery of PT symmetry in the KH instability provides a new perspective in the study of classical instabilities in conservative systems. The PT-symmetry analysis for matrix differential operators developed in the present study is applicable to a broader range of systems. We expect that all classical conservative systems are PT-symmetric, and spontaneous PT-symmetry breaking is a generic mechanism for the onset of instabilities in these systems.
\begin{acknowledgments}
This research was supported by the U.S. Department of Energy (DE-AC02-09CH11466). We thank Hongxuan Zhu for fruitful discussions.
\end{acknowledgments}

\bibliographystyle{apsrev4-1}
\bibliography{PT_KH.bbl}

\end{document}